\documentclass[prl,twocolumn,a4paper,nofootinbib]{revtex4-1}

\usepackage{amsmath}
\usepackage{amssymb}
\usepackage{amsthm}
\usepackage{amsfonts}
\newtheorem{theorem}{Theorem}
\newtheorem{lem}{Lemma}

\begin{document}

\title{Braiding defects in topological stabiliser codes of any dimension cannot be universal}
\author{Paul Webster}
\affiliation{Centre for Engineered Quantum Systems, School of Physics, The University of Sydney, Sydney, NSW 2006, Australia}
\author{Stephen D.~Bartlett}
\affiliation{Centre for Engineered Quantum Systems, School of Physics, The University of Sydney, Sydney, NSW 2006, Australia}

\date{10 August 2020}

\begin{abstract}
Braiding defects in topological stabiliser codes has been widely studied as a promising approach to fault-tolerant quantum computing. We present a no-go theorem that places very strong limitations on the potential of such schemes for universal fault-tolerant quantum computing in any spatial dimension. In particular, we show that, for the natural encoding of quantum information in defects in topological stabiliser codes, the set of logical operators implementable by braiding defects is contained in the Clifford group. Indeed, we show that this remains true even when supplemented with locality-preserving logical operators.
\end{abstract}

\maketitle
Topological stabiliser codes are a promising approach to protecting quantum information, as they possess high thresholds and allow for the correction of very general local errors.  In addition, any quantum logic gates that can be performed with locality-preserving logical operators (a generalisation of transversal operators) are fault-tolerant, meaning that local errors remain local and correctable. Quantum computing requires a universal set of fault-tolerant logic gates, however, and the set of locality-preserving logical operators in a topological stabiliser code cannot be universal \cite{Eastin,Bravyi,Pastawski,Webster}.

An alternative approach to performing fault-tolerant logic in topological codes involves braiding \emph{topological defects}.  Defects can have a richer braiding structure compared with the anyonic excitations of the topological code.  A number of schemes to encode qubits into defects and implement gates by braiding these defects have been proposed and explored for codes in two spatial dimensions.  One class of encodings uses holes in topological stabilizer codes~\cite{Raussendorf1,Raussendorf2,BombinHole,Fowler1,Fowler2,Fowler3,Hastings,Brown,Brell}.  Such codes can admit entangling gates by braiding, and schemes based on supplementing these gates by locality-preserving logical operators have been proposed for both the two-dimensional surface and colour codes~\cite{Fowler1,Fowler2};  however, braiding holes alone cannot give a universal set of logical operators in these codes~\cite{Escobar}.  Focussing on more exotic topological defects, encodings into twist defects that lie on the boundary of domain walls have also been widely studied due to their similarity to non-abelian anyons \cite{Brown,Bombin,Bombin2,Kesselring,Cong,Yoder,Scruby,BarkeshliGenon1,BarkeshliGenon2,BarkeshliTheory1,BarkeshliTheory2}. Fault-tolerant logic with such twist defects has been studied in a range of two dimensional codes, including the surface code~\cite{Brown,Bombin}, colour code~\cite{Kesselring}, subsystem colour codes~\cite{Bombin2} and the $\mathbb{Z}_3$ quantum double model~\cite{Cong} (a generalisation of the surface code to qutrits). As with holes, braiding twist defects in these codes does not give a universal gate set.

Because braiding defects has many similar features to braiding non-abelian anyons even in abelian topological models~\cite{Bombin,Brown}, and because some non-abelian braiding models allow for universal gate sets~\cite{Freedman,Kitaev}, it is natural to question whether there exist (abelian) topological stabiliser codes with defects that similarly allow for universality through braiding. The idea seems plausible, as there are examples of topological models that do not admit a universal set of fault-tolerant logical operators, but which do allow for universality when genons (a type of twist defect) are introduced and braided~\cite{BarkeshliGenon1}.

In this Letter, we prove that braiding defects in topological stabiliser codes in any dimension cannot be universal. Specifically, we show that the set of operators obtained by braiding topological defects in any topological stabiliser code of any spatial dimension is contained in the Clifford group. We show that this constraint remains even when such operators are supplemented by locality-preserving logical operators. Our result applies to any code in which logical qubits are encoded into pairs of topological defects -- of any type -- with states of the qubit defined by the excitations present at the defects, and a process to distinguish these states by braiding another excitation around each defect. Such an encoding is the natural approach using defects, generalising two-dimensional topological stabiliser code defect schemes including those using holes \cite{Raussendorf1,Raussendorf2,BombinHole,Fowler1,Fowler2,Fowler3,Hastings,Brown} and pairs of domain walls with twists \cite{Brown,Bombin,Bombin2,Kesselring,Scruby} while also allowing for a range of phenomena specific to higher dimensional codes. In our companion paper \cite{WebsterLong}, we show that even if we permit more exotic encodings, our main result -- that braiding defects in topological stabiliser codes cannot be universal -- remains unchanged.

\textit{Defects in Topological Stabiliser Codes.---}A topological stabiliser code is a stabiliser code in which physical qubits are arranged in a lattice of ${D}\geq 2$ spatial dimensions such that the stabiliser group admits a local generating set and logical information is encoded in topological degrees of freedom  \cite{Bravyi,Pastawski,Webster,Kitaev2,BravyiKitaev,Bombin3,Bombin5,Bombin4,TerhalRev}. To allow for a notion of defects, we assume a fundamental (defect-free) code that is translationally invariant. A defect is defined to be a $k$-dimensional region of this code where translational invariance is broken, with $0\leq k<{D}$. 

Two dimensional topological stabiliser codes admit anyons -- localised excitations that may be distinguished by their fusion and braiding rules. In higher dimensional codes, anyons generalise to the class of \textit{topological excitations}. These are excitations that admit an analogous notion of braiding. Specifically, there must be at least two spatial dimensions in which such excitations are localised and be freely propagated by a locality-preserving unitary \cite{WebsterLong}.

We refer to a defect at which topological excitations can condense (i.e.~be created by local operators at the defect) as a \emph{topological defect}, as this condensation allows the defect to carry topological charge. (To ensure that the defects are fundamentally topological, topological excitations are the only objects we allow to condense.) For example, in the surface code, holes created with rough or smooth boundaries can condense electric or magnetic charge, respectively, and so are topological defects.  Transparent domain walls, such as a lattice dislocation in the surface code, form topological defects referred to as \textit{twists} at their boundaries, again allowing excitations to condense. Note that the entire boundary of the domain wall, for example a pair of points for a one dimensional wall, or a loop for a two dimensional wall, is considered to be a single topological defect. 

Because topological defects allow for topological excitations to condense, they can be used to encode protected quantum information. Here, we consider encodings such that each logical qubit is encoded in a pair of topological defects, with the computational basis states of the qubit corresponding to the presence or absence of a topological excitation, $a$, localised at one or other of the defects. (For clarity, we assume that the relevant excitation fuses with a copy of itself to return to the vacuum, so that the encoding is of a qubit, but this can easily be generalised to more general excitations and qudits). We also assume that there exists a topological excitation $b$ that gives a phase of $-1$ when braided with $a$. This ensures that it is possible to distinguish the computational basis states, by braiding $b$ around one of the defects.

This encoding has the important property that its logical Pauli operators are \textit{topological locality-preserving logical operators} (TLPLOs) \cite{WebsterLong}. A TLPLO is a logical operator that admits an implementation by which one or more topological excitations are created, propagated along a non-trivial path through the code and then annihilated. Specifically, the logical $\bar{X}$ operator for an encoded qubit corresponds to the creation of the topological excitation $a$, its propagation through the code and then its annihilation at the other encoding defect. The $\bar{Z}$ operator corresponds to the braiding of topological excitation $b$ around either of the defects. The property that logical Pauli operators are TLPLOs is shared by all standard approaches to encoding information in topological stabiliser codes (with \cite{Raussendorf1,Raussendorf2,BombinHole,Fowler1,Fowler2,Fowler3,Hastings,Brown,Bombin,Bombin2,Kesselring,Scruby} or without \cite{Webster,Kitaev2,BravyiKitaev,Bombin3,Bombin5,Bombin4,TerhalRev} defects) in all spatial dimensions. 

Encoding logical information into defects in this way naturally generalises two dimensional defect encodings, while also allowing for the range of exotic features that arise only in higher dimensional codes to be manifested. In particular, we note the following possibilities that are included in the setup we consider, but are not possible in two dimensional codes:
\begin{enumerate}
\item Defects do not need to be point-like, but can include objects such as loop-like or line-like twist defects \cite{WebsterLong};
\item Defects may be associated with non-Pauli Hamiltonian terms. For example we allow for domain walls that arise at the boundary of non-Clifford locality-preserving logical operators in higher dimensional code, such as that in the three-dimensional colour code corresponding to the transversal $\bar{T}$ gate, which has non-Pauli $XS$-type Hamiltonian terms \cite{Yoshida,Ni};
\item The topological excitations that condense at the defects can include non-eigenstate excitations such as SPT-excitations \cite{Yoshida} or linking charges \cite{BombinLinking}, and we allow for TLPLOs that are implemented by the creation, propagation and annihilation of such excitations;
\item The underlying code can exhibit fracton order \cite{Vijay,Dua}, and the excitations of limited mobility characteristic of such models, such as planons, may condense at defects. We refer the reader to our companion paper \cite{WebsterLong} for an example.
\end{enumerate}

With such encodings, we have two distinct ways to implement fault-tolerant logical operations. The first is to apply a locality-preserving logical operator; these are unitary operators that map local (and hence correctable) errors to local  errors. We note that TLPLOs are a subset of locality-preserving logical operators, as the propagation of topological excitations (unlike defects) can be performed using locality-preserving unitary operators~\cite{WebsterLong}. Specifically, $\bar{X}$ operators correspond to a product of operators applied to a manifold connecting a pair of defects, and $\bar{Z}$ operators correspond to a product of operators applied to a topologically non-trivial manifold enclosing a defect. Indeed, the group of TLPLOs is a normal subgroup of the group of locality-preserving logical operators, since locality-preserving logical operators map TLPLOs to TLPLOs \cite{WebsterLong}. Strong constraints on the locality-preserving logical operators admitted by a topological stabiliser code are known~\cite{Bravyi,Pastawski}, as is a full classification for a large class of the most widely studied instances~\cite{Webster}.

Alternatively, fault-tolerant logical operators may be implemented by \textit{braiding} defects. We define such a \emph{braiding logical operator} to be a process of code deformation \cite{Brown,Bombin5}, by which the set of qubits on which one or more defects have support are smoothly changed over time, such that the final defect configuration is indistinguishable from the initial configuration. We require that the setup remains a valid defect encoding at all times throughout the process. Specifically, the distance of the encoding (i.e.~minimum weight of a logical Pauli operator) must be maintained up to a constant. This ensures that the protection against local errors provided by the initial defect encoding is maintained throughout, and hence that braiding logical operators are fault-tolerant.

\textit{Results.---}We now present our main result (Theorem \ref{Theorem1}), that combinations of braiding and locality-preserving logical operators are contained in the Clifford group. We prove this result with two lemmas:  first, showing that braiding logical operators map TLPLOs to TLPLOs (Lemma \ref{Lem2}) and second, that all TLPLOs are logical Pauli operators (Lemma \ref{Lem1}).

\begin{lem}
Let $\bar{B}$ be a braiding logical operator, and $\bar{U}$ be a TLPLO. Then, $\bar{B}\bar{U}\bar{B}^\dag$ is a TLPLO. \label{Lem2}
\end{lem}
\begin{proof}
$\bar{B}$ can be decomposed into a product of code deformations, $\bar{B}=\prod_{i=1}^n B_i$, with a sequence of codes $\{\mathcal{C}_i\}_{0\leq i \leq n}$ (where $\mathcal{C}_0=\mathcal{C}_n$) such that $B_i: \mathcal{C}_{i-1} \to \mathcal{C}_i$. Since braiding is smooth, we may assume that each $B_i$ acts non-trivially only on a local region around the topological defects of code $\mathcal{C}_{i-1}$. Let $\bar{U}_k:\mathcal{C}_k \to \mathcal{C}_k$ (for $0\leq k \leq n$) be defined by $\bar{U}_k= \left(\prod_{i=1}^{k} B_i\right) \bar{U} \left(\prod_{i=1}^{k} B_i\right)^\dag$. Noting that $\bar{U}_0=\bar{U}$ is a TLPLO, we proceed by induction to show that $U_k$ is a TLPLO for all $k\leq n$.

Assume that $\bar{U}_{k-1}$ is a TLPLO. Then, since topological excitations can be freely propagated (within a subspace of dimension at least two), $\bar{U}_{k-1}$ can be deformed to an equivalent TLPLO, $\bar{U}'_{k-1}$ such that the support of $\bar{U}'_{k-1}$ and the support of $B_k$ intersect only in local regions around topological defects at which topological excitations condense in the implementation of $\bar{U}_{k-1}$. We note that since topological defects allow topological excitations to condense, the implementation of $\bar{U}'_{k-1}$ may require the condensation of additional topological excitations, but is nonetheless still implementable by the condensation and propagation of some set of topological excitations and is thus a TLPLO.
Now, $B_k$ acts trivially on $\bar{U}'_{k-1}$ except for in local regions around topological defects at which topological excitations condense in the implementation of $\bar{U}'_{k-1}$. Thus, $\bar{U}_k$ is equivalent up to local operators to a TLPLO implemented by the condensation of topological excitations at the corresponding topological defects to those at which topological excitations condensed in the implementation of $\bar{U}'_{k-1}$ and the same propagation of topological excitations as $\bar{U}'_{k-1}$ elsewhere. Since any two logical operators that are equivalent up to local operators must have equivalent logical action, this implies that $\bar{U}_k$ is equivalent to this TLPLO.

Hence, $\bar{U}_k$ is a TLPLO for all $k\leq n$ and so, in particular, $\bar{U}_n=\bar{B}\bar{U}\bar{B}^\dag$ is a TLPLO.
\end{proof}

\begin{lem}
All TLPLOs are logical Pauli operators. \label{Lem1}
\end{lem}
\begin{proof}
Let $\bar{U}$ be a TLPLO in a $D$-dimensional code. We show first that $\bar{U}$ must be contained in the Clifford hierarchy (indeed in the $D$th level of the hierarchy) and then proceed to show it must in fact be in the Pauli group by induction on the Clifford hierarchy.

If $\bar{U}$ is outside the $D$th level of the Clifford hierarchy, then there exists a set of $D$ Pauli TLPLOs, $\{\bar{P}_1,\ldots \bar{P}_{D}\}$, whose sequential group commutator with $\bar{U}$ is non-trivial \cite{Pastawski,Yoshida2}. This implies that the corresponding defect-free code with periodic boundary conditions has a corresponding set of $D+1$ TLPLOs implemented by propagating the same topological excitations as for $\bar{U}$ and $\{\bar{P}_1,\ldots \bar{P}_{D}\}$ around a topologically non-trivial path whose sequential group commutator is also non-trivial. This implies that the defect-free code has a locality-preserving logical operator outside the $D$th level of the Clifford hierarchy \cite{Pastawski,Yoshida2}, which contradicts Ref.~\cite{Bravyi}.

We now prove the result by induction on the Clifford hierarchy. Specifically, assume all TLPLOs in the $k$th level of the Clifford hierarchy are logical Pauli operators. Then, if $\bar{U}$ is in the $(k+1)$th level of the Clifford hierarchy, its commutator with all single qubit logical Pauli operators must be logical Pauli operators, and so $\bar{U}$ must be in the Clifford group. We thus need only to prove the case where $\bar{U}$ is a Clifford TLPLO.

In this case, for an arbitrary logical qubit, $j$, encoded in a pair of topological defects, $\mathcal{D}$ and $\mathcal{D}'$, there exist logical Pauli operators $\bar{P}= \alpha\prod_i \bar{X}_i^{a_i}\bar{Z}_i^{b_i}$ and $\bar{Q}=\beta\prod_i \bar{X}_i^{c_i}\bar{Z}_i^{d_i}$ (for $\alpha,\beta \in \mathbb{C}$ and $a_i,b_i,c_i,d_i \in \mathbb{Z}_2$) such that $P=\bar{X}_j\bar{U}\bar{X}_j\bar{U}^\dag$ and $Q=\bar{Z}_j\bar{U}\bar{Z}_j\bar{U}$. Since the support of the group commutator of two operators must be contained in their intersection, this implies that $\text{supp}(\bar{P}) \subseteq \text{supp}(\bar{X}_j) \cap \text{supp}(\bar{U})$ and $\text{supp}(\bar{Q}) \subseteq \text{supp}(\bar{Z}_j) \cap \text{supp}(\bar{U})$

However, by construction, $\text{supp}(\bar{X}_j)$ does not enclose any defect in a topologically non-trivial manifold and $\text{supp}(\bar{Z}_j)$ does not connect any defect pairs, so $b_i=c_i=0$ for all $i$. Also, $\text{supp}(\bar{X}_j)$ does not connect any pair of defects except for $\mathcal{D}$ and $\mathcal{D}'$ and $\text{supp}(\bar{Z})$ also does not enclose any pair of defects except for $\mathcal{D}$ or $\mathcal{D}'$, so $a_i=d_i=0$ for all $i\neq j$. Additionally, since $\bar{U}$ is a TLPLO, it corresponds to the path of a topological excitation that necessarily can be freely moved in some direction perpendicular to that in which $\mathcal{D}$ and $\mathcal{D}'$ are separated and to some direction in which $\mathcal{D}$ is localised. Thus, for any particular implementations of $\bar{X}_j$ and $\bar{Z}_j$ we can deform $\bar{U}$ such that $\text{supp}(\bar{P}) \subseteq\text{supp}(\bar{X}_j) \cap \text{supp}(\bar{U})$ does not connect defects $\mathcal{D}$ and $\mathcal{D}'$ and $ \text{supp}(\bar{Q}) \subseteq\text{supp}(\bar{Z}_j) \cap \text{supp}(\bar{U})$ does not enclose $\mathcal{D}$ or $\mathcal{D}'$. Hence, $a_j=d_j=0$.

Thus, $\bar{P}=\alpha\bar{I}$ and $\bar{Q}=\beta\bar{I}$ and so the commutators of $\bar{U}$ with all single qubit logical operators is trivial up to a phase. Thus, $\bar{U}$ is a logical Pauli operator.
\end{proof}
 
We now present and prove our main result.

\begin{theorem}
The set of logical operators implementable by products of braiding and locality-preserving logical operator on logical qubits encoded in defects is contained in the Clifford group. \label{Theorem1}
\end{theorem}
\begin{proof}
By Lemma \ref{Lem2}, if $\bar{B}$ is a braiding logical operator and $\bar{U}$ is a TLPLO, then $\bar{B}\bar{U}\bar{B}^\dag$ is a TLPLO. Also, as noted above, if $\bar{L}$ is a locality-preserving logical operator then $\bar{L}U\bar{L}^\dag$ is a TLPLO. If $\bar{A}$ is a product of braiding and locality-preserving logical operators, it acts on $\bar{U}$ by a sequence of conjugations by braiding or locality-preserving logical operators. Thus, $\bar{A}$ maps TLPLOs to TLPLOs under conjugation. By Lemma \ref{Lem1}, the group of TLPLOs is contained in the logical Pauli group and, as all logical Pauli operators are TLPLOs, this implies that the group of TLPLOs is equal to the logical Pauli group. Hence, $\bar{A}$ maps logical Pauli operators to logical Pauli operators under conjugation, and so is contained in the Clifford group.
\end{proof}

\textit{Discussion.---} Theorem \ref{Theorem1} shows that the set of operators implementable by products of braiding and locality-preserving logical operators is highly restricted. In particular, we note that since the Clifford group is finite, it cannot be universal for quantum computing. Schemes based on braiding defects in two dimensional topological stabiliser codes to implement the Clifford group (such as with twists in the surface code \cite{Bombin,Brown}) have long been known, and our result shows that such schemes cannot be improved upon within this framework. 

While this conclusion is consistent with previous results in two dimensional topological stabiliser codes \cite{Bravyi}, it is surprisingly restrictive in higher dimensions. Specifically, we note that for any spatial dimension $D\geq 2$, there exist $D$-dimensional topological stabiliser codes (without defects) that admit locality-preserving logical operators that are strictly in the $D$th level of the Clifford hierarchy \cite{Kubica,Webster}. One might expect that an analogous result would hold for braiding defects -- that a corresponding defect scheme could be constructed in a $D$-dimensional code that allowed a braiding logical operator in the $D$th level of the Clifford hierarchy. This expectation is furthered by the greater range of defects and braiding phenomena possible in higher dimensional codes~\cite{Mesaros,Else,Wang,Yoshida}. However, our result shows that all braiding logical operators are confined to the second level of the Clifford hierarchy, and hence that this is not possible.

The restrictiveness of our results can be understood to reflect a tradeoff between locality-preserving logical operators and braiding logical operators. Specifically, because braiding logical operators permute the group of TLPLOs (as shown in Lemma \ref{Lem2}), the structure of this group constrains the power of braiding in a model. In particular, in our framework the group of TLPLOs is the logical Pauli group (as shown in Lemma \ref{Lem1}), and it is this structure that restricts braiding logical operators to the Clifford group. Interestingly, this tradeoff is reminscent of that identified for gates implemented by braiding anyons in two-dimensional topological quantum field theories in Ref.~\cite{Beverland}. Indeed, that work observed that models that admit the Pauli group as locality-preserving logical operators have operators implemented by braiding anyons contained in the Clifford group, analogous to our result.

Our result is highly restrictive on achieving universal fault-tolerant quantum computing with such codes, but we can nonetheless consider how our no-go result may be circumvented. Considering the identified tradeoff, one approach is to sacrifice braiding logical operators to allow for a larger group of locality-preserving logical operators. This is realised by the standard encoding of logical qubits into defect-free, higher-dimensional topological stabiliser codes. Such an encoding does not permit braiding, but allows for non-Clifford locality-preserving logical operators \cite{Bravyi,Kubica,Webster}, which we have shown are not admitted by qubits encoded in defects. However, it is known that the group of locality-preserving logical operators in a topological stabiliser code cannot be universal \cite{Bravyi}.

Alternatively, one may seek a more general approach to encoding qubits in defects that is optimised for braiding logical operators, at the expense of locality-preserving logical operators. In particular, the assumption that all logical $\bar{Z}$ operators must be implementable as TLPLOs could be removed. However, this is challenging for two reasons. First, because a sufficient (but not necessary) condition for $\bar{Z}$ to be a TLPLO is that the topological excitation used to define the computational basis for the encoding is an eigenstate excitation, such an approach requires direct use of non-eigenstate excitations to define the encoding of logical qubits, which has no precedent in topological stabiliser codes with or without defects. Second, without a locality-preserving implementation of $\bar{Z}$, the application and measurement of logical Pauli operators would be significantly more challenging, that would make such a scheme incompatible with standard approaches to quantum computing which assume the simplicity of such operations. In our companion paper \cite{WebsterLong}, we consider defect schemes in which we remove this assumption, and allow for encodings into any number of defects. We find that the set of gates implementable by combinations of braiding and locality-preserving logical operators still cannot be universal, even in this highly generalised framework.

We must conclude that braiding defects is insufficient to allow us to overcome the limitations of no-go theorems such as that of Eastin-Knill \cite{Eastin} and Bravyi-Konig \cite{Bravyi}. To achieve universality in topological codes, we must incorporate additional techniques such as magic state distillation \cite{MSD}, stabiliser state injection \cite{YoderBacon,Vasmer,WebsterLong}, dimensional jumping \cite{BombinDJ}, just-in-time decoding \cite{BombinJIT,BrownJIT} or topological charge measurement \cite{Cong}. Such techniques require adaptive implementation of logical operators and so are outside the scope of all our results \cite{BombinSingleShot,WebsterLong}.

\begin{acknowledgments}
This work is supported by the ARC via projects CE170100009 and DP170103073. PW acknowledges support from The University of Sydney Nano Institute via The John Makepeace Bennett Gift.
\end{acknowledgments}

\end{document}